\documentclass[aps,prx,reprint,twocolum,amsmath,superscriptaddress,amssymb,showpacs,nofootinbib]{revtex4-2}

\usepackage[utf8]{inputenc}
\usepackage[english]{babel}
\usepackage{booktabs} 
\usepackage{comment} 
\usepackage{hyperref} 
\usepackage{physics} 

\usepackage[table,xcdraw,dvipsnames]{xcolor} 

\allowdisplaybreaks

\usepackage{lipsum}
\usepackage{xfrac}

 \usepackage{multirow}

\usepackage{amsmath,verbatim,latexsym,amssymb,indentfirst,mathrsfs,mathtools,amsthm,bbm,bm,hyperref,url,cancel,subcaption}
\usepackage[font=small,labelfont=bf,format=plain,justification=raggedright,singlelinecheck=false]{caption}
\usepackage{verbatim,indentfirst}
\usepackage[title,titletoc]{appendix}

\makeatletter
\def\th@mytheorem{%
  \normalfont 
  \thm@headfont{\itshape} 
  \thm@notefont{\normalfont\itshape}%
  \thm@preskip\parindent
  \thm@postskip\thm@preskip
  \thm@headpunct{.} 
}
\makeatother

\theoremstyle{mytheorem}
\newtheorem{theorem}{Theorem}

\usepackage{graphicx}
\graphicspath{{images/}}
\usepackage{epstopdf}

\makeatletter
\def\p@subsection{}
\makeatother
\makeatletter
\def\p@subsubsection{}
\makeatother

\makeatletter
\newcommand\footnoteref[1]{\protected@xdef\@thefnmark{\ref{#1}}\@footnotemark}
\makeatother

\usepackage{soul}

\usepackage{color}
\usepackage[dvipsnames]{xcolor}
\usepackage[normalem]{ulem}

\newcommand{\Dim}{d}   

\newcommand{\Sites}{N}  

\newcommand{\tot}{ {\rm tot} }
\def\id{\mathbbm{1}}   

\newcommand{\Hil}{\mathcal{H}}  

\newcommand{\JParen}{ {(j)} }

\newcommand*{\Set}[1]{\left\{  #1  \right\}}

\renewcommand\th{ {\rm th} }

\begin{document}

\title{Noncommuting charges can remove non-stationary quantum many-body dynamics}

\author{Shayan Majidy} 
\email{smajidy@fas.harvard.edu}
\affiliation{Department of Physics, Harvard University, Cambridge, Massachusetts 02139, USA}
\affiliation{Institute for Quantum Computing, University of Waterloo, Waterloo, Ontario N2L 3G1, Canada}

\date{\today}

\begin{abstract}
Studying noncommuting conserved quantities, or 'charges,' has revealed a conceptual puzzle: noncommuting charges hinder thermalization in some ways yet promote it in others. While many quantum systems thermalize according to the Eigenstate Thermalization Hypothesis (ETH), systems with 'dynamical symmetries' violate the ETH and exhibit non-stationary dynamics, preventing them from equilibrating, much less thermalizing.  We demonstrate that each pair of dynamical symmetries corresponds to a specific charge. We find that introducing new charges that do not commute with existing ones disrupts these symmetries, thereby eliminating non-stationary dynamics and facilitating thermalization. We illustrate this behavior across various models, including the Hubbard model and Heisenberg spin chains. Our findings demonstrate that noncommuting charges can enhance thermalization by reducing the number of local observables that thermalize according to the ETH.
\end{abstract}

\maketitle

\section{Introduction}\label{sec:introduction}


How can isolated quantum many-body systems thermalize? This question is largely answered by the eigenstate thermalization hypothesis (ETH)~\cite{srednicki1994chaos, deutsch1991quantum}. Consider a set of local observables $O_i$, an initial state $|\psi\rangle$ governed by a Hamiltonian $H$, and the corresponding expectation values $\langle O_i(t) \rangle$. A system is said to have thermalized at time $t$ if $\langle O_i(t) \rangle \approx \langle O_i \rangle_{\text{th}}:=\text{tr}[\rho_{\text{th}} O_i]$ for each $O_i$, where $\rho_{\text{th}}$ is the thermal state with the temperature determined by the energy of $|\psi\rangle$. Most quantum many-body systems thermalize according to the ETH~\cite{ueda2020quantum, rigol2008thermalization}. Integrable, superintegrable~\cite{fagotti2014conservation}, and many-body localized (MBL) systems~\cite{abanin2019colloquium}, for example, violate the ETH by evolving to stationary but non-thermal states. Systems with quantum scars~\cite{serbyn2021quantum} or dynamical symmetries~\cite{buvca2019non} defy the ETH with persistent non-stationary dynamics.


Most studies of quantum thermalization assume that systems' charges commute.  This assumption segments the system’s Hilbert space into `charge sectors,' which help in determining when a state thermalizes~\cite{d2016quantum} and its form. However, charges do not always commute. In fact, noncommutation of observables is a fundamental aspect of quantum theory, impacting uncertainty relations~\cite{sen2014uncertainty, coles2017entropic}, measurement disturbance~\cite{fewster2020quantum, polo2022detector}, and foundational tests~\cite{aspect1999bell, emary2013leggett, majidy2019exploration, majidy2021detecting}. Discarding this assumption leads to the breakdown of derivations of the thermal state's form~\cite{nyh2018beyond, nyh2016microcanonical, guryanova2016thermodynamics, lostaglio2017thermodynamic}, the ETH~\cite{murthy2022non}, Onsager relations~\cite{manzano2022non}, and other results~\cite{majidy2023noncommuting}.  Exploring the implications of noncommuting charges has spurred the development of a dynamic subfield, bridging thermodynamics and many-body physics~\cite{majidy2023noncommuting, majidy2024effects, nyh2016microcanonical, guryanova2016thermodynamics, lostaglio2017thermodynamic}.


A conceptual puzzle in studying noncommuting charges relates to their impact on thermalization. Recent results suggests that noncommuting charges may inhibit thermalization by reducing entropy production rates~\cite{upadhyaya2023happens, manzano2018Squeezed}, imposing more stringent constraints on unitary operations compared to commuting charges~\cite{marvian2022restrictions}, and causing larger deviations from the thermal state under physically plausible conditions~\cite{murthy2022non}. However, there is also evidence that noncommuting charges could facilitate thermalization. For instance, they are known to increase average entanglement levels~\cite{majidy2023non, majidy2023critical}, crucial for thermalization in closed quantum systems, and destabilize MBL~\cite{potter2016symmetry}.


This puzzle is theoretically compelling and has potential implications for quantum technologies. The primary challenge in scaling quantum computers is mitigating decoherence~\cite{majidy2024building, preskill1998reliable, gottesman2010introduction, campbell2017roads}, which is closely related to thermalization~\cite{popovic2023thermodynamics}. Discovering controllable obstacles to thermalization may imply obstacles to certain forms of decoherence.  Systems with noncommuting charges, like spin systems~\cite{nyh2020noncommuting, nyh2022how, murthy2022non, majidy2023non} and squeezed states~\cite{manzano2022non, manzano2018Squeezed}, are already used in quantum computing methods, like quantum dots~\cite{loss1998quantum} and optical techniques~\cite{zhong2020quantum}. This potential is further supported by recent advances showing that non-Abelian symmetric operations are universal for quantum computing~\cite{rudolph2023two, freedman2021symmetry}.


A promising approach to solving this puzzle is to introduce noncommuting charges into systems that initially resist thermalization and monitor the resulting changes. This approach was effectively used in demonstrating that noncommuting charges destabilize MBL~\cite{potter2016symmetry}. In this work, we move from non-thermalizing to non-stationary dynamics by introducing noncommuting charges into systems with dynamical symmetries.  In this work we focus on local dynamics. 

Dynamical symmetries, proposed by Bu{\v{c}}a et al.~\cite{buvca2019non, medenjak2020isolated, medenjak2020rigorous, chinzei2020time}, are operators $A$ that satisfy $\comm{H}{A} = \lambda A$ with $\lambda \in \mathbb{R}$ and $\lambda \neq 0$, where $H$ is the system's Hamiltonian. Any local operator $\mathcal{O}_i$ overlapping with $A$ (i.e., $\tr[\mathcal{O}_iA] \neq 0$) will not thermalize, displaying non-stationary dynamics in violation of the ETH. Dynamical symmetries always come in pairs, exist in open systems~\cite{buvca2019non, tindall2020quantum, buvca2019dissipation}, and can be responsible for the non-stationary dynamics observed in quantum time crystals~\cite{medenjak2020isolated}, OTO crystals~\cite{buvca2022out}, and quantum attractors~\cite{buca2020quantum}. Dynamical symmetries depart from the conventional definition of `symmetry' as an invariance under a transformation. Therefore, Noether's theorem does not establish the link between dynamical symmetries and charges.

Our main result demonstrates that introducing noncommuting charges into a system can disrupt its existing dynamical symmetries. Starting with a system that possesses dynamical symmetries, we identify charges from the dynamical symmetries and then impose additional charges. If the added charge commutes with the existing physical charges, the dynamical symmetries remain intact. However, if the added charge does not commute with the existing charges, the original dynamical symmetries are eliminated. This work highlights a fundamental conflict between noncommuting charges and dynamical symmetries, revealing another mechanism by which noncommuting charges facilitate thermalization.


\section{Results}

\subsection{Set-up} \label{Sec:Background}

Consider a closed quantum system consisting of a lattice with $N$ sites. Each site corresponds to a Hilbert space $\Hil$ of finite dimensionality $\Dim$. The system is governed by a Hamiltonian $H = \sum_k E_k \dyad{\psi_k}$, where $\ket{\psi_k}$ are energy eigenstates with energies $E_k$. The time dependent state $\ket{\Phi(t)} =  \sum_k \exp(-iE_kt)c_k \ket{\psi_k}$ will have fixed total energy $E = \mel{\Phi(t)}{H}{\Phi(t)}$, where we set $\hbar = 1$. The expectation value of an observable $\mathcal{O}$ for the state $\ket{\Phi(t)}$ is
\begin{align}
    \expval{\mathcal{O}(t)} &= \sum_{j,k} e^{-i(E_k-E_j)t} c_j^*c_k \mel{\psi_j}{\mathcal{O}}{\psi_k} \label{eq:expvalop1}.
\end{align}
$\expval{\mathcal{O}}_{\rm{th}} := \tr[\rho_{\rm th}O]$ is the thermal expectation value where $\rho_{\rm th}$ is the thermal state with temperature fixed by $E$. If an out-of-equilibrium $\ket{\Phi(0)}$ thermalizes, we expect $\expval{\mathcal{O}(t)}$ to approach $\expval{\mathcal{O}}_{\rm{th}}$,
\begin{equation}
   \expval{\mathcal{O}(t \rightarrow \infty)} = \expval{\mathcal{O}}_{\rm{th}} + C(t),
\end{equation}
where $C(t)$ accounts for fluctuations. This second term is also what allows $\expval{\mathcal{O}(0)}$ to exhibit non-stationary dynamics. The ETH gives a set of conditions for which this expectation holds~\cite{srednicki1994chaos, deutsch1991quantum}.

Systems with dynamical symmetries violate the ETH~\cite{buvca2019non, medenjak2020rigorous}. Since dynamical symmetries result in time-dependent dynamics, they are violating the off-diagonal ETH. We denote a dynamical symmetry by a non-Hermitian operator $A$, such that 
\begin{equation}
    \comm{H}{A} = \lambda A,  
\end{equation}
where $\lambda$ is a real and non-zero constant. We assume that $A$ is extensive by restricting it to have the form $A = \sum_{j=1}^{N} \tilde{A}^\JParen,$ where $\tilde{A}^{\JParen}$ is an operator that acts non-trivially on one site and as the identity on all other sites. Furthermore, for every dynamical symmetry $A$, there exists another $A^{\dagger}$, $\comm{H}{A^{\dagger}}= - \lambda A^{\dagger}$. Thus, dynamical symmetries always come in pairs. A system can have multiple pairs of dynamical symmetries. In that case, we add subscripts as follows $\comm{H}{A_{\beta}} = \lambda_{\beta}A_{\beta}$. Furthermore, we define $A_{+\beta} := A_{\beta}$ and $A_{-\beta}:= A_{\beta}^{\dagger}$. The expectation value of an observable $\mathcal{O}$ that overlaps with one dynamical symmetry, $\tr[A \mathcal{O}] \neq 0$, will continue to change over time. As proven in Ref.~\cite{buvca2019non}, dynamical symmetries are sufficient conditions for non-stationary dynamics. For certain locally interacting systems, they are also necessary conditions~\cite{buvca2023unified}.

We illustrate dynamical symmetries with an example. Consider a collection of spin-$\frac{1}{2}$ particles in an external field. The system's Hamiltonian is $H = B\sum_{j}\sigma_z^{(j)}$ where $\sigma_{x,y,z}$ are the Pauli matrices and the superscript denotes the operator acts non-trivially on the $j$th site. Such a Hamiltonian has dynamical symmetries $A_{\pm} = \sum_{j} (\sigma_x^{(j)} \pm i \sigma_y^{(j)})$. The operator $\sigma_{x}^{(j)}$, for example, overlaps with $A_{\pm}$ and will thus experience non-stationary dynamics. The persistent recurrence from precessing spins prevents thermalization of $\sigma_{x}^{(j)}$. This is intuitive---an isolated spin in a magnetic field will continually process around the applied field. This very simple example also demonstrates the variation in the magnitude of the non-stationary dynamics and how it depends on the initial state. Consider one extreme case; all the spins are aligned and orthogonal to the external field. The magnitude of oscillations will be at its maximum. In the other extreme case, all the spins are aligned or anti-aligned with the external field. This is an eigenstate, and we thus have no dynamics. This is a trivial example of having ``non-thermalizing'' dynamics. 

Finally, we introduce noncommuting charges. Charges are Hermitian operators that commute with the Hamiltonian, $\comm{H}{Q} = 0$. A system can have multiple charges that we distinguish with Greek letter subscripts, and these charges can be noncommuting, $\comm{Q_{\alpha}}{Q_{\beta}} \neq 0$. We add a third condition to hermiticity and commutation for physically motivated reasons. Let $\tilde{Q}_\alpha$ denote a Hermitian operator defined on $\Hil$. We denote by $\tilde{Q}_\alpha^\JParen$ the $\tilde{Q}_\alpha$ defined on the $j^\th$ subsystem's $\Hil$. We denote an extensive observable 
\begin{align}
   Q_\alpha  \coloneqq  \sum_{j = 1}^{\Sites} \tilde{Q}_\alpha^\JParen \equiv   \sum_{j = 1}^\Sites \id^{\otimes (j - 1)}  \otimes \tilde{Q}_\alpha^\JParen  \otimes \id^{\otimes (\Sites - j) }. \label{eq:charges}
\end{align}
Our third condition is that the charges are extensive. While this work focuses on closed quantum systems, noncommuting charges and dynamical symmetries can also exist in open quantum systems \cite{zhang2020stationary, buvca2019non, tindall2020quantum, buvca2019dissipation}.

\subsection{Correspondence between charges and dynamical symmetries} \label{Sec:Correspondence}

In this subsection, we present a correspondence between the existence of charges and dynamical symmetries. This correspondence is in the form of two theorems. One theorem identifies a charge from pairs of dynamical symmetries. The other shows that for a class of commuting charges, one can always find a Hamiltonian that conserves those charges and has the corresponding dynamical symmetries. 


Charges $Q_{\alpha}$ that generate Lie algebras are important in physics because they describe conserved quantities such as angular momentum, particle number, electric charge, color charge, and weak isospin~\cite{majidy2023noncommuting, das2014lie,iachello06lie,gilmore12lie}, i.e., everything in the Standard Model of particle physics. The Lie algebras relevant to studying noncommuting charges are finite-dimensional because we study systems with a finite number of linearly independent charges~\cite{nyh2022how}. The algebras are also defined over the complex number because the operators are Hermitian. Finally, the algebras are semisimple so that the operator representation of the charges can be diagonalized (not necessarily simultaneously diagonalized) and permit a special basis that we introduce later \cite{humphreys2012introduction}. Many physically significant algebras satisfy these conditions, such as $\mathfrak{su}(N)$, $\mathfrak{so}(N)$, and all simple Lie algebras. From this point onward, we denote by $\mathfrak{g}$ a finite-dimensional semisimple complex Lie algebra. 

Essential to our study are Cartan-Weyl bases~\cite{humphreys2012introduction}, which we will now introduce. An algebra's dimension $c$ equals the number of generators in a basis for the algebra. The algebra's rank $r$ is the dimension of the algebra's maximal commuting subalgebra, the largest subalgebra in which all elements are commuting. For example, consider the usual basis for $\mathfrak{su}$(2)---the Pauli-operators $\Set{  \sigma_x,  \sigma_y,  \sigma_z  }$. There are three generators in this basis, so $c=3$, and none of these operators commute with one another, so $r=1$. A Cartan subalgebra is a maximal Abelian subalgebra consisting of semisimple elements, $h \in \mathfrak{h}$. From this point onward, we denote by $\mathfrak{h}$ a Cartan subalgebra of a $\mathfrak{g}$. Every $\mathfrak{g}$ will have a $\mathfrak{h}$. The $\mathfrak{g}$ definition includes a vector space $V$ defined over a field $F$. A form is a map $V \times V \rightarrow F$. The Killing form of operators $x, y \in \mathfrak{g}$ is the bilinear form $(x,y):= \tr({\rm ad}x \cdot {\rm ad}y)$ where ad$x$ is the image of $x$ under the adjoint representation of $\mathfrak{g}$. Let $\beta(h) \coloneqq (h_{\beta}',h)$ where $h, h_{\beta}' \in \mathfrak{h}$. $\beta(h)$ is a root of $\mathfrak{g}$ relative to $\mathfrak{h}$ if there exists a non-zero operator $X_{\beta} \in \mathfrak{g}$ such that 
\begin{equation}
    \comm{h}{X_{\beta}} = \beta(h) X_{\beta}. \label{eq:rootvectors}   
\end{equation}
These operators $X_{\beta}$ are called root vectors. Denote by $\Delta$ all roots of $\mathfrak{g}$ with respect to $\mathfrak{h}$. If $\beta \in \Delta$, then so is $-\beta$. Thus, root vectors always come in pairs $X_{\pm \beta}$. Finally, one can identify different Cartan-Weyl bases for each $\mathfrak{g}$. A Cartan-Weyl basis consists of a Cartan subalgebra $\mathfrak{h}$ and one or more pairs of root vectors $X_{\pm \beta}$. The choice of Cartan-Weyl is not unique.

We highlight additional features of root vectors which are important to this work. For all root vectors $\comm{X_{+\beta}}{X_{-{\beta}}} \in \mathfrak{h}$ (Proposition 7.17 of Ref.~\cite{hall2013lie}). $X_{\pm \beta}$ can always be chosen such that $X_{+\beta} = X_{-\beta}^{\dagger}$ (p.~273 of Ref.~\cite{campoamor2019group}). It follows then that 
\begin{equation}
    \left(\comm{X_{\beta}}{X_{-{\beta}}}\right)^{\dagger} = \comm{X_{\beta}}{X_{-{\beta}}}\label{eq:DS1}
\end{equation}
and that
\begin{equation}
    \comm{X_{+\beta}}{\comm{\mathcal{H}}{X_{-\beta}}}= \left(\comm{X_{-\beta}}{\comm{\mathcal{H}}{X_{+\beta}}}\right)^{\dagger}\label{eq:DS2}
\end{equation}
where $\mathcal{H}$ is any Hermitian operator. 

\begin{theorem} \label{th:charges}
    For every pair of dynamical symmetries $A_{\pm \beta}$ that a Hamiltonian has, there exists a charge $Q_{\beta} =  \comm{A_{+ \beta}}{A_{- \beta}}$. 
\end{theorem}
\begin{proof}
    For $Q_{\beta}$ to be a charge, it must be conserved by the Hamiltonian, a Hermitian operator, and extensive. Using the Jacobi identity, Eq.~\eqref{eq:DS1}, and Eq.~\eqref{eq:DS2}, we find that
    \begin{align}
        \comm{H}{\comm{A_{+\beta}}{A_{-\beta}}} &= \comm{A_{+\beta}}{\comm{H}{A_{-\beta}}} - \comm{A_{-\beta}}{\comm{H}{A_{+\beta}}}\\
        &= -\lambda_{\beta} \left(\comm{A_{+\beta}}{A_{-\beta}} + \comm{A_{-\beta}}{A_{+\beta}} \right)\\
        &= 0.
    \end{align}
    Thus, the Hamiltonian conserves $Q_{\beta}$. For $Q_{\beta}$ to be a charge, it must also be Hermitian, which it is
    \begin{align}
        Q_{\beta}^{\dagger}
        &= \left(A_{\beta}A_{\beta}^{\dagger} - A_{\beta}^{\dagger}A_{\beta} \right)^{\dagger}\\ &=  \left( A_{\beta}A_{\beta}^{\dagger} - A_{\beta}^{\dagger}A_{\beta} \right) = Q_{\beta}.
    \end{align}
   Finally, recall that we are studying dynamical symmetries, which are $k=1$ local. Thus, all $Q_{\beta} =  \comm{A_{+ \beta}}{A_{- \beta}}$ will also be $1$-local and are charges by our definition.
\end{proof}

We say a set of dynamical symmetries `produces' a set of charges $\{Q_{\beta}\}$ when the span of $\{Q_{\beta}\}$ equals the span of the set of charges found using Theorem \ref{th:charges}. A note on notation: We use the subscript $\beta$ when considering charges defined by dynamical symmetries (as in Theorem 1) and $\alpha$ otherwise.

\begin{theorem} \label{th:DS}
For every set of charges $\{Q_{\alpha}\}$ that form a $\mathfrak{h}$ of $\mathfrak{g}$, there exists a Hamiltonian $H$ conserving those charges such that the root vectors that complete the Lie algebra are dynamical symmetries of $H$.
\end{theorem}

\begin{proof}
    Consider a Hamiltonian of the form $H = H_{\mathfrak{g}} + \sum_{\alpha} c_{\alpha} Q_{\alpha}$, where $H_{\mathfrak{g}}$ is a Hamiltonian that commutes with all charges that generate $\mathfrak{g}$ and $c_{\alpha}$ are constants. The charges that generate $\mathfrak{g}$ form a basis for the algebra~\cite{nyh2022how}. Since the root vectors can be written in this basis, $\comm{H_{\mathfrak{g}}}{X_{{\pm}\beta}} = 0$. For any charge and root vector $\comm{Q_{\alpha}}{X_{\pm \beta}} = \beta(\alpha) X_{\pm \beta}$. It follows then that $X_{\pm \beta}$ are dynamical symmetries:
    \begin{equation}
        \comm{H}{X_{\pm \beta}} = \sum_{\alpha} c_{\alpha} \comm{Q_{\alpha}}{X_{\pm \beta}}  = \Big(\sum_{\alpha} c_{\alpha}  \beta(\alpha)\Big) X_{\pm \beta}.
    \end{equation}
\end{proof}

We say a set of charges `can produce' a set of dynamical symmetries $\{A_{\pm \beta}\}$, when $\{A_{\pm \beta}\}$ equals one set of dynamical symmetries that can be found using Theorem \ref{th:DS}. Theorem \ref{th:DS} is for a set of charges that form a $\mathfrak{h}$ of $\mathfrak{g}$. Any set of charges that generate a $\mathfrak{g}$ can be partitioned into $\frac{c}{r}$ sets of mutually commuting charges~\cite{nyh2022how}.

We pause here to clarify one point. Our objects of study are Hamiltonians $H$ that have dynamical symmetries $\{A_{\pm \beta}\}$. We are asking what happens to these dynamical symmetries when we enforce noncommuting charges onto the Hamiltonian. Thus, we first need to identify a charge of the system from its dynamical symmetries. This can be done with Theorem 1, $Q_{\beta} = \comm{A_{+\beta}}{A_{-\beta}}$. We now modify the Hamiltonian such that it conserves charges $\{Q_{\alpha} \}$ that do not commute with each other or the existing charges $\{Q_{\beta}\}$. We then ask what happens to the $\{A_{\pm \beta}\}$ when we make this change. The final Hamiltonian's noncommuting charges will form a Lie algebra. The initial Hamiltonian's dynamical symmetries and original charges will also form a Lie algebra. In both cases, these can be the same Lie algebra, but the Hamiltonians have clearly different charges. This point is further clarified by the examples below.


A simple setting to illustrate this correspondence is the Hubbard model. We choose this model for various reasons. First, it has been shown to demonstrate non-stationary behaviour emerging from dynamical symmetries and the forms of the charges and dynamical symmetries are known~\cite{buvca2019non}. Additionally, its two commuting charges form separate $\mathfrak{h} $'s, unique from the other examples with commuting charges we study. Finally, it is a physically important model---the prototypical model of strongly correlated materials.

Consider a chain of $N$ fermions. Denote by $c_{\sigma}^{(j)\dagger}$ and $c_{\sigma}^{(j)}$ the creation and annihilation operators for a fermion of spin $\sigma$ at lattice site $j$, $\sigma \in \{ \uparrow, \downarrow\}$. Denote by $n^{(j)}_{\sigma} \coloneqq c_{\sigma}^{(j)\dagger} c_{\sigma}^{(j)}$ the number operator for fermions of spin $\sigma$ at lattice site $j$. The 1D Hubbard model's Hamiltonian can be written as~\cite{tasaki1998hubbard}
\begin{align}
    H =& \sum_{j=1}^{N-1} \sum_{\sigma=\uparrow,\downarrow} \Big[ -t\left(c_{\sigma}^{(j)\dagger} c_{\sigma}^{(j+1)} + c_{\sigma}^{(j+1)\dagger} c^{(j)}_{\sigma}\right) + U n_{\uparrow}^{(j)}n_{\downarrow}^{(j)} \nonumber \\ & - \mu \left(n_{\uparrow}^{(j)} + n_{\downarrow}^{(j)}\right) + \frac{B}{2} \left(n_{\uparrow}^{(j)} - n_{\downarrow}^{(j)}\right) \Big],
\end{align}
where $t$ is the hopping amplitude, $U$ is the on-site Coulomb interaction, $\mu$ is the chemical potential, and $B$ is the strength of a constant external magnetic field. The first term describes the kinetic energy of electrons hopping between neighbouring sites. The second term describes the Coulomb repulsion between two electrons on the same site.  The third term adjusts the total number of electrons in the system. The fourth term splits the energy levels of the up-spin and down-spin electrons.

The Hubbard model has two pairs of dynamical symmetries~\cite[Supplementary Materials]{buvca2019non}. The first pair are
\begin{align}
     S_{+ z}^{\rm tot} = \sum_{j = 1}^{L}  c_{\uparrow}^{(j)\dagger} c_{\downarrow}^{(j)} \quad \text{and} \quad  S_{- z}^{\rm tot} = \sum_{j = 1}^{L} c_{\downarrow}^{(j)\dagger} c_{\uparrow}^{(j)},
\end{align}
and the second pair are
\begin{align}
    \eta^{\rm tot}_{+z}  
    = \sum_{j = 1}^{N}  (-1)^{j}c^{(j)\dagger}_{\uparrow}c^{(j)\dagger}_{\downarrow} \; \; \text{and} \; \;
    \eta^{\rm tot}_{-z} 
    = \sum_{j = 1}^{L} (-1)^{j}c_{\downarrow}^{(j)}c_{\uparrow}^{(j)}.
\end{align}
Using theorem \ref{th:charges}, we identify two charges:
\begin{align}
 \comm{ S_{+ z}^{\rm tot}}{ S_{- z}^{\rm tot}} = S_{z}^{\rm tot} &= \sum_{j=1}^{L} \left(n_{\uparrow}^{(j)} - n_{\downarrow}^{(j)}\right), \, \text{and} \\
    \comm{\eta^{\rm tot}_{+z}}{ \eta^{\rm tot}_{-z}} =  \eta_{z}^{\rm tot} &= \sum_{j=1}^{L}(n_{ \uparrow}^{(j)} + n_{\downarrow}^{(j)} - 1).\label{eta}
\end{align}
These charges are the two known charges of the system. 

Starting from the charges, we can also identify the dynamical symmetries. When $B=0$ and $\mu=0$, the Hubbard Hamiltonian contains two sets of charges that generate $\mathfrak{su}(2)$~\cite{jakubczyk20122}.  For $B \neq 0$ and $\mu \neq 0$, the Hubbard model has two sets of charges that generate Cartan subalgebras of $\mathfrak{su}(2)$. Each charge, $S_{z}^{\rm tot}$ and $\eta_{z}^{\rm tot}$, is an element in one of these algebras. We can use these Cartan subalgebras to complete a Cartan-Weyl basis for $\mathfrak{su}(2)$. Doing so, we find  two sets of generators $\{S_{z}^{\rm tot},  S_{+z}^{\rm tot}, S_{-z}^{\rm tot} \}$ and $\{\eta_{z}^{\rm tot},  \eta_{+z}^{\rm tot}, \eta_{-z}^{\rm tot} \}$. This demonstrates how Cartan-Weyl bases can be used to identify the dynamical symmetries from the charges.

\subsection{Noncommuting charges' effect on dynamical symmetries}\label{sec:Effect}

In this subsection, we consider the following setting. We begin with a system that has dynamical symmetries $\{A_{\pm \alpha'}\}$. From Theorem \ref{th:charges}, we can identify its charges $\{Q_{\alpha}\}$, and those charges and dynamical symmetries will form a Lie algebra $\mathfrak{g}$. Furthermore, from Theorem \ref{th:DS}, we know that a Hamiltonian with those charges and dynamical symmetries exists. This system experiences non-stationary dynamics in all observables $\mathcal{O}_i$ that overlap with any elements in $\{A_{\pm \alpha'}\}$. We then introduce one or more charges into the system that do not commute with the existing charges. This introduction of charge(s) removes the existing non-stationary dynamics. Finally, we explain why this effect occurs with noncommuting charges.

First we illustrate this procedure using charges that generate a Cartan subalgebra of $\mathfrak{su}(2)$ and to charges that generate $\mathfrak{su}(2)$. We then do the same analysis for the $\mathfrak{su}(3)$ algebra. These sections further demonstrate the difference in introducing commuting and noncommuting charges. Each of these subsections also presents Hamiltonians that can be used to explore the distinction between commuting and noncommuting charges.

The Hubbard model illustrates that in systems with commuting charges, the dynamical symmetries of each charge can coexist. This section will present another example with two commuting charges and two examples with noncommuting charges. A summary of the examples studied in this work is presented in Table.~\ref{tab:summary}. 

\begin{table*}
    \begin{center}
    \begin{tabular}{@{}lccccc@{}}
    \toprule
     & \begin{tabular}[c]{@{}c@{}}Subalgebra of $\mathfrak{su}(2)$\\ (e.g., XXX model + field)\end{tabular} & \begin{tabular}[c]{@{}c@{}}2 subalgebras of $\mathfrak{su}(2)$\\ (e.g., Hubbard model)\end{tabular} & \begin{tabular}[c]{@{}c@{}}Subalgebra of $\mathfrak{su}(3)$\\ (e.g., Eq.~\eqref{eq:su3eg})\end{tabular} & \begin{tabular}[c]{@{}c@{}}Full $\mathfrak{su}(2)$\\ (e.g., XXX model)\end{tabular} & \begin{tabular}[c]{@{}c@{}}Full $\mathfrak{su}(2)$\\ (e.g., Eq.~\eqref{eq:su3eg})\end{tabular} \\ \midrule
    \begin{tabular}[c]{@{}l@{}}Charges'\\ symmetry\end{tabular} & Abelian & Abelian & Abelian & Non-Abelian & Non-Abelian \\
    Charges & $S_z^{\mathrm{tot}}$ & $S_z^{\mathrm{tot}}, \eta_z^{\mathrm{tot}}$ & $\tau_3^{\mathrm{tot}}, \tau_8^{\mathrm{tot}}$ & $S_{\alpha=x,y,z}^{\mathrm{tot}}$ & $\tau_{\alpha=1,\ldots, 8}^{\mathrm{tot}}$ \\
    \begin{tabular}[c]{@{}l@{}}Dynamical\\ symmetries\end{tabular} & $S_{\pm z}^{\mathrm{tot}}$ & $S_{\pm z}^{\mathrm{tot}}, \eta_{\pm z}^{\mathrm{tot}}$ & $A_{\pm 1}^{\mathrm{tot}}, A_{\pm 2}^{\mathrm{tot}}, A_{\pm 3}^{\mathrm{tot}}$ &  &  \\ \bottomrule
    \end{tabular}
    \caption{\textbf{Summary of the examples studied in this work.} Throughout this work, we study five examples of systems with charges and with or without dynamical symmetries. This table summarizes our main result---that introducing noncommuting charges into systems with dynamical symmetries removes those symmetries. One explanation for this result is that a Cartan-Weyl basis (consisting of commuting charges and dynamical symmetries) and noncommuting charges each form a basis for the algebra being studied. The Hamiltonian, thus, cannot commute with all of the noncommuting charges and not commute with some elements of the Cartan-Weyl basis.}
    \label{tab:summary}
    \end{center}
\end{table*}


Dynamical symmetries always come in pairs, and a corresponding charge exists for every pair (Theorem \ref{th:charges}). In the simplest case, a single charge introduces a U(1) symmetry. Thus, this dynamical symmetry pair and the charge will form an $\mathfrak{su}(2)$ Lie algebra. This is the simplest example to consider. Consider again a system of $N$ sites. We denote by $\sigma_{\alpha}^{\JParen}$ a Pauli operator acting on the $j$th site. We define the components of the spin-$\tfrac{1}{2}$ angular momentum operators as $S_\alpha^\tot \coloneqq \sum_{j = 1}^{\Sites}  \sigma_{\alpha}^\JParen$. The system begins with one charge that forms a Cartan subalgebra of $\mathfrak{su}(2)$. To be concrete, we choose this element to be $S_{z}^\tot$. Next, we complete the Cartan-Weyl basis by identifying the root vectors of the algebra, which are
\begin{equation}
    S_{\pm z}^{\rm tot} = \sum_{j=1}^{N}     \id^{\otimes (j - 1)}  \otimes S_{\pm z}^\JParen  \otimes 
   \id^{\otimes (\Sites - j) } \equiv   \sum_{j = 1}^{\Sites}  S_{\pm z}^\JParen \label{eq:spinladders}
\end{equation}
where $S_{\pm z} = \frac{1}{2} ( \sigma_x  \pm  i  \sigma_y ) $. Like with the Hubbard model, it is straightforward to reverse this procedure. Doing so, we check that $\comm{S_{+z}}{S_{-z}} = S_{z}$, and thus $\comm{S_{+z}^{\rm tot}}{S_{-z}^{\rm tot}} = S_{ z}^{\rm tot}$.

We now introduce another charge that does not commute with the existing one. However, a Hamiltonian that conserves two $\mathfrak{su}(2)$ charges will necessarily conserve all three \cite{nyh2022how}. Thus, we introduce two additional charges into the system. To complete $\mathfrak{su}(2)$, we include $S_{x}^{\rm tot}$ and $S_{y}^{\rm tot}$ as charges. Introducing the charges means that the Hamiltonian now commutes with $S_{\pm x}^{\rm tot}$ and $S_{\pm y}^{\rm tot}$, and thus also commutes with $S_{\pm z}^{\rm tot}$. Together, the three conservation laws eliminate the original dynamical symmetries. When the system had one charge of $\mathfrak{su}(2)$, it had two dynamical symmetries. Now that the system has three charges of $\mathfrak{su}(2)$ it has no dynamical symmetries. This example contrasts with the Hubbard model, where the dynamical symmetries of $S_{\pm z}^{\rm tot}$ could coexist with a charge that commutes with $S_{z}^{\rm tot}$.

Various Hamiltonians exist that can naturally transition from a single charge forming a Cartan subalgebra of $\mathfrak{su}(2)$ to three charges constituting the full $\mathfrak{su}(2)$. An example of this is the Heisenberg model under an external field.
\begin{align}
    H_{2} =&\, \frac{B}{2}\Big(\sum_{j=1}^{N} \sigma_z^{\JParen} \Big) + \frac{J}{2} \Big(\sum_{\langle j, k \rangle } \sum_{\langle\langle j, k \rangle\rangle }  \sigma_x^{(j)}\sigma_x^{(k)} + \sigma_y^{(j)}\sigma_y^{(k)}  \nonumber\\ & + \sigma_z^{(j)}\sigma_z^{(k)}\Big) ,\label{eq:XXXHeis1}
\end{align}
where $\langle j, k \rangle$ indicates the sum over nearest neighbours, $\langle \langle j, k \rangle \rangle$ indicates the sum is over next-nearest neighbours, $B$ is the strength of an external magnetic field, and $J$ is a coupling constant. For $B \neq 0$, the system has one charge corresponding to a Cartan subalgebra of $\mathfrak{su}(2)$ and one pair of dynamical symmetries. By setting $B=0$, we introduce two additional noncommuting charges into the system, thereby removing the dynamical symmetries. We included the next-nearest neighbour interaction to break integrability. Alternatively, we could construct our Hamiltonian using genuine three-body interactions that are SU(2)-symmetric, such as
\begin{align}
    \sigma_x^{(j)}\sigma_y^{(j+1)}\sigma_z^{(j+2)} + \sigma_y^{(j)}\sigma_z^{(j+1)}\sigma_x^{(j+2)} + \sigma_z^{(j)}\sigma_x^{(j+1)}\sigma_y^{(j+2)} \nonumber \\
    - \sigma_z^{(j)}\sigma_y^{(j+1)}\sigma_x^{(j+2)} - \sigma_x^{(j)}\sigma_z^{(j+1)}\sigma_y^{(j+2)} - \sigma_y^{(j)}\sigma_x^{(j+1)}\sigma_z^{(j+2)},
\end{align}
and still break the symmetry with an external field.


In our study, charges can commute in two ways. First, they may be components of distinct algebras, as exemplified by the Hubbard model. Alternatively, they can belong to the same Cartan subalgebra. To demonstrate this second way charges can commute, we turn to $\mathfrak{su}(3)$. $\mathfrak{su}(3)$ has dimension $c= 8$ and rank $r = 2$. Thus, a system with an $\mathfrak{su}(3)$ symmetry has eight charges that can be partitioned into $\frac{c}{r} = 4$ sets of mutually commuting charges. These sets generate Cartan subalgebras. The eight charges of $\mathfrak{su}(3)$ can be represented by the Gell--Mann matrices~\cite{Cahn06Semi}, $\tau_i$ for $i=1$ to $8$. 

We begin with one Cartan sublagebra of $\mathfrak{su}(3)$. For example, take the Cartan subalgebra with $ \tau_3$ and $\tau_8$. In the three-dimensional representation of $\mathfrak{su}$(3), these operators can be represented with
\begin{align}
   \tau_3  = \begin{bmatrix}
   1 & 0 & 0 \\
   0 & -1 & 0 \\
   0 & 0 & 0
   \end{bmatrix}
   \quad \text{and} \quad
   \tau_8  =  \frac{1}{\sqrt{3}}  \begin{bmatrix}
   1 & 0 & 0 \\
   0 & 1 & 0 \\
   0 & 0 & -2
   \end{bmatrix}.
\end{align}
As before, our charges will be sums over these operators on each site, $Q_{1} = \sum_{j=1}^N \tau_{3}^{\JParen}$ and $Q_{2} = \sum_{j=1}^N \tau_{8}^{\JParen}$. Using this Cartan subalgebra, we construct a Cartan-Weyl basis. This requires identifying $\frac{c-r}{2} = 3$ pairs of root vectors.
We define the following operators, $A_{+1} \coloneqq \tau_1 + i \tau_2$, $A_{+2} \coloneqq \tau_4 + i \tau_5$, and $A_{+3} \coloneqq \tau_6 + i \tau_7$. These operators and their Hermitian conjugates are the root vectors. As we did in the previous example, we can construct the dynamical symmetries by taking sums over the operators on each site in our chain, $A_{\pm \beta}^{\rm tot} := \sum_{j=1}^{N} A_{\pm \beta}^{\JParen}$. This is our initial set-up: a system with two commuting charges and six dynamical symmetries. 

Since we begin with three pairs of dynamical symmetries, one may expect we would have three commuting charges. Consider the operators,
\begin{align}
    Q_{1} &= c_{1} \comm{A_{+1}}{A_{-1}} = \frac{1}{\sqrt{2}}
   \begin{bmatrix}
      1 & 0 & 0 \\
      0 & -1 & 0 \\
      0 & 0 & 0
   \end{bmatrix},\\
    Q_{\tilde{2}} &=  c_{2} \comm{A_{+2}}{A_{-2}} = \frac{1}{\sqrt{2}}
   \begin{bmatrix}
      1 & 0 & 0 \\
      0 & 0 & 0 \\
      0 & 0 & -1
   \end{bmatrix}, \, \text{and}\\
   Q_{\tilde{3}} &=  c_{3} \comm{A_{+3}}{A_{-3}} = \frac{1}{\sqrt{2}}
   \begin{bmatrix}
      0 & 0 & 0 \\
      0 & 1 & 0 \\
      0 & 0 & -1
   \end{bmatrix}.
\end{align}
The important point here is that all three are not linearly independent. One can see this from the Killing forms between all pairs not being zero: $(Q_1, Q_{\tilde{2}}) = 3, (Q_1, Q_{\tilde{3}}) = -3$, and $(Q_{\tilde{2}}, Q_{\tilde{3}}) = 3$. A Cartan subalgebra basis will require two charges. Thus, we sum over two of these three operators. We are free to do so in different ways. The choice that recovers the original two operators we considered is summing over charges $Q_{\tilde{2}}$ and $Q_{\tilde{3}}$: $Q_{2} = \frac{1}{\sqrt{3}}\left(Q_{\tilde{2}} + Q_{\tilde{3}}\right)$.

We now want to introduce other charges that do not commute with the existing ones, i.e., more of the charges that generate $\mathfrak{su}(3)$. However, recall the dynamical symmetries for $\tau_3$ and $\tau_8$ are linear combinations of the other six Gell--Mann matrices. If the Hamiltonian now commutes with any elements from other Cartan subalgebras of $\mathfrak{su}(3)$, it will stop some combination of $A_{\pm 1}^{\rm tot}$, $A_{\pm 2}^{\rm tot}$, and $A_{\pm 3}^{\rm tot}$ from being dynamical symmetries. If the Hamiltonian commutes with all of $\mathfrak{su}(3) $'s charges, it will nullify the system's dynamical symmetries. When the system had two commuting charges, it had six dynamical symmetries. Now that it has eight noncommuting charges, it has no dynamical symmetries.

Similar to the $\mathfrak{su}(2)$ example, we can study Hamiltonians for the $\mathfrak{su}(3)$ example. These Hamiltonians are less familiar but can be derived using the procedure in Ref.~\cite{nyh2022how},
\begin{align}\label{eq:su3eg}
    H_3 =&\, \frac{J}{2}\Big(\sum_{\alpha}  \sum_{\langle j, k \rangle } \sum_{\langle\langle j, k \rangle\rangle }  \tau_{\alpha}^{(j)}\tau_{\alpha}^{(k)}\Big) + \frac{B_1}{2} \Big(\sum_j \tau_{3}^{(j)} \Big) \nonumber \\
    & + \frac{B_2}{2} \Big(\sum_j \tau_{8}^{(j)} \Big).
\end{align}
Setting $B_1$ and $B_2$ to zero allows the noncommuting charges to give way to dynamical symmetries.

Finally, we can ask why noncommuting charges remove dynamical symmetries. A Cartan-Weyl basis (consisting of commuting charges and dynamical symmetries) and noncommuting charges each form a basis for the algebras being studied. Thus, the elements of the Cartan-Weyl basis can be expressed in terms of the noncommuting charges, i.e., the dynamical symmetries can be written as a linear combination of the noncommuting charges. Therefore, if the Hamiltonian commutes with more noncommuting charges, it will commute with fewer dynamical symmetries. This is why introducing noncommuting charges leads to the removal of existing dynamical symmetries.

\section{Discussion} \label{Sec:OutLook}

This study adds a key piece to the puzzle of whether noncommuting charges aid or hinder thermalization~\cite{majidy2023non, majidy2024effects, majidy2023noncommuting, kranzl2022experimental, murthy2022non, noh2023eigenstate, siegl2023imperfect, ares2023lack, tabanera2023thermalization, marvian2023theory, dabholkar2024ergodic, upadhyaya2023happens, majidy2023noncommuting, garcia2024estimation}. We found that introducing noncommuting charges into a system can remove its existing dynamical symmetries, thereby removing non-stationary dynamics. Noncommuting charges thus are able to promote thermalization by reducing the number of local observables that thermalize according to the Eigenstate Thermalization Hypothesis. 

Non-Abelian symmetries have now been shown to remove three forms of non-thermalizing behaviour. The main result of this work shows the removal of dynamical symmetries and Ref.~\cite{potter2016symmetry} shows destabilizng MBL. Reference \cite{o2020tunnels} presents a procedure for constructing Hamiltonians with quantum scars by taking a non-Abelian symmetric Hamiltonian and adding a perturbation that breaks the non-Abelian symmetry. Reversing this procedure, one finds a case where introducing introducing non-Abelian symmetries removes quantum scars. Unlike the quantum scar work, our analysis targets local observables and can extend beyond Hamiltonians to Lindbladians since noncommuting charges and dynamical symmetries are relevant in both closed and open systems~\cite{zhang2020stationary, buvca2019non, tindall2020quantum, buvca2019dissipation}. 

One promising opportunity for future work is unifying these three results. One motivation to look into this is that the Hamiltonians found in Ref.~\cite{o2020tunnels} and those studied in this work overlap, suggesting a potential connection. A second reason comes from the entanglement in the systems studied in these three results. Reference~\cite{protopopov2017effect} shows that, in contrast to typical systems where MBL occurs, SU(2) symmetry mandates eigenstates with entanglement exceeding area law. Meanwhile, Reference~\cite{o2020tunnels} finds that breaking SU(2) symmetry removes states with subthermal entanglement entropy. Both of these works align with recent research that indicates non-Abelian symmetries increase entanglement entropy~\cite{majidy2023non, majidy2023critical}. Thus, it seems promising to look into the underlying mechanism causing non-Abelian symmetries to increase entanglement and understand its connection to supressing non-thermalizing phenomena.

A separate direction is to study the effect of breaking non-Abelian symmetries without introducing dynamical symmetries. For example, one could alter the Heisenberg Hamiltonian to break its non-Abelian symmetry by adding a coupling term like $(\sigma_x^{(j)}\sigma_y^{(j+1)} - \sigma_y^{(j)}\sigma_x^{(j+1)})$. If noncommuting charges diminish non-thermalizing dynamics under these conditions, their role in promoting thermalization would be further substantiated. If not, this would suggest that the thermalization-promoting ability of noncommuting charges is limited to certain systems. A third opportunity is to explore dynamical symmetries without any restrictions on locality and study their effect on the thermalization of less local observables, as done in Ref.~\cite{medenjak2020isolated}.

The primary aim of this work is to understand the influence of noncommuting charges on thermalization. Several studies indicate that noncommuting charges obstruct thermalization, as evidenced by deviations from the thermal state~\cite{murthy2022non}, slower entropy production~\cite{manzano2022non}, more constrained dynamics compared to systems with commuting charges~\cite{marvian2022restrictions}, and difficulties in characterizing the thermal state~\cite{nyh2018beyond, nyh2016microcanonical, guryanova2016thermodynamics, lostaglio2017thermodynamic}. Conversely, there is evidence suggesting that noncommuting charges facilitate thermalization, demonstrated by increased average entanglement~\cite{majidy2023non, majidy2023critical} and the elimination of many-body localization, quantum scars, and dynamical symmetries~\cite{potter2016symmetry, protopopov2017effect, o2020tunnels}. One possible resolution is that noncommuting charges impede thermalization, while simultaneously enhancing equilibration by removing obstacles to reaching stationary states. Further work is needed to understand if this is true and to elucidate the nature of the relationship between noncommuting charges and thermalization.

\section{Data Availability}

All materials necessary to support the conclusions of this study are included in the main article.

\bibliography{apssamp}

\begin{acknowledgments}

SM would like to thank Berislav Bu\v{c}a and Abhinav Prem for useful discussions and Nicole Yunger Halpern for asking him a question that ultimately motivated this work---what do dynamical symmetries have to do with noncommuting charges? Further thanks are extended to Matthew Graydon, Raymond Laflamme, Zachary Mann, Jos\'e Polo-G\'omez, Abhinav Prem, and Nicole Yunger Halpern for their comments on the manuscript. This work received support from the Vanier C.G.S. 

\end{acknowledgments}

\end{document}